%% file: main.tex
\documentclass[11pt]{article}

\input{preamble}

\algblockdefx{MRepeat}{EndRepeat}{\textbf{repeat}}{}
\algnotext{EndRepeat}
\begin{document}
	\input{titlepage}

	\input{introduction}

	\input{preliminaries}

	\input{certificate}

	\input{algorithm}

	\input{discussion}

	\bibliographystyle{alpha}
	\bibliography{references}
	\input{appendix}

\end{document}

%% file: preamble.tex
\usepackage{amsthm}
\usepackage{amsmath}
\usepackage{amssymb}
\usepackage{mathtools}
\usepackage{tikz}
\usepackage{thmtools}
\usepackage{fullpage}
\usepackage{url}
\usepackage{hyperref}
\usepackage{color}
\usepackage{graphicx}
\usepackage[capitalise]{cleveref}
\usepackage{algorithm}
\usepackage[noend]{algpseudocode}
\usepackage{bbm}

\newcommand{\eps}{\ensuremath{\varepsilon}}

\newcommand{\Ind}{\mathbbm{1}}

\DeclareMathOperator*{\E}{\mathbb{E}}

\declaretheorem[]{theorem}
\declaretheorem[sibling=theorem]{definition}
\declaretheorem[sibling=theorem]{lemma}

\declaretheorem[sibling=theorem]{claim}

\newtheorem*{theorem*}{Open question}

\newenvironment{prevproof}[2]{\noindent {\bf {Proof of {#1}~\ref{#2}.}}}{$\hfill\qed$\vskip \belowdisplayskip}

%% file: titlepage.tex
\title{
	Fast Modular Subset Sum using Linear Sketching 
}

\author{
  Kyriakos Axiotis\\
  MIT\\
  \texttt{kaxiotis@mit.edu} 
  \and
  Arturs Backurs\\
  MIT\\
  \texttt{backurs@mit.edu} 
  \and
  Christos Tzamos\\
  University of Wisconsin-Madison\\
  \texttt{tzamos@wisc.edu}
}
\date{}

\maketitle

\begin{abstract}
Given $n$ positive integers, the \emph{Modular Subset Sum} problem asks if a subset adds up to a given target $t$ modulo a given integer $m$.
This is a natural generalization of the Subset Sum problem (where $m=+\infty$) with ties to additive combinatorics and cryptography.

Recently, in \cite{Bringmann17, KX17}, efficient algorithms have been developed for the non-modular case, running in near-linear pseudo-polynomial time.
For the modular case, however, the best known algorithm by Koiliaris and Xu \cite{KX17} runs in time $\widetilde{O}\left(m^{5/4}\right)$.

In this paper, we present an algorithm running in time $\widetilde{O}(m)$, which matches a recent conditional lower bound 
of \cite{ABHS17}
based on the Strong Exponential Time Hypothesis.
Interestingly, in contrast to most previous results on Subset Sum, our algorithm does not use the Fast Fourier Transform. Instead, 
it is able to simulate the ``textbook'' Dynamic Programming algorithm much faster, using ideas from linear sketching.
This is one of the first
applications of sketching-based techniques to obtain fast algorithms for combinatorial problems in an offline setting.


\end{abstract}

%% file: introduction.tex
\section{Introduction}

In the \emph{Subset Sum} problem, one is given a multiset of integers and an integer target $t$ and is asked to decide
if there exists a subset of the integers that sums to the target $t$.
Subset Sum is a classic problem known to be NP-complete, originally included as one of Karp's 21 NP-complete problems \cite{Karp72}.
Despite its NP-completeness, it is possible to obtain algorithms that are pseudo-polynomial in the target $t$.
In particular, the ``textbook'' Dynamic Programming algorithm of Bellman \cite{Bellman57} solves the problem in $O(nt)$ time.

Due to its importance and applications in various areas, there have been a lot of works 
improving the runtime \cite{Pisinger99, Pferschy99, Pisinger03, KX17, Bringmann17}, 
obtaining polynomial space \cite{LN10, Bringmann17}, or 
achieving polynomial decision tree complexity \cite{Meyer84, CIO15, ES16, KLM18}.
In addition to Subset Sum, there has recently been a lot of effort in obtaining faster algorithms for the more general problem
of \emph{Knapsack} \cite{EW18, BHSS18, AT18}.

The most recent result by Bringmann \cite{Bringmann17} brings down the runtime for Subset Sum to $\widetilde{O}(t)$,
which is known to be optimal assuming the Strong Exponential Time Hypothesis \cite{ABHS17}. 
The fastest known deterministic algorithm by Koiliaris and Xu has a runtime of $\widetilde{O}(\sqrt{n} t)$ \cite{KX17}.

An important generalization of Subset Sum
is the \emph{Modular Subset Sum} problem, in which sums are taken over the finite cyclic group $\mathbb{Z}_m$ for some given integer $m$.
This problem and its structural properties 
has been studied extensively in Additive Combinatorics \cite{EGZ61, Olson68, Szemeredi70, Olson75, Alon87, HLS07, Vu08}.
The trivial algorithm for deciding whether a given target is achievable modulo $m$ runs in time $O(nm)$.
Interestingly, even the fastest algorithm for (non-Modular) Subset Sum of \cite{Bringmann17} does not give any nontrivial
improvement over this runtime. 
Koiliaris and Xu \cite{KX17} were able to obtain an algorithm running in time $\widetilde{O}(m^{5/4})$ by exploiting
structural properties of the Modular Subset Sum problem implied by Additive Combinatorics \cite{HLS07}.

The main contribution of our work is an optimal algorithm for the Modular Subset Sum problem.

\subsection{Our Contributions}

In this paper, we present an algorithm for the Modular Subset Sum problem 
running in time $\widetilde{O}(m)$\footnote{
	If $n > m$ our algorithm would need to spend $\Omega(n)$ time just to read the input. However, if the input is represented succinctly, our algorithm runs in $\widetilde{O}(m)$ time even if $n > m$. An $\widetilde{O}(m)$-size succinct representation for a multiset 
of a universe with $m$ elements is always possible by listing the elements and their multiplicities.
For this reason, we omit the dependence on $n$. See discussion in Section~\ref{discussion}.}. 

\begin{theorem}
There is an $\widetilde{O}(m)$-time algorithm that with high probability returns all subset sums that are attainable,
i.e. it solves the Modular Subset Sum problem for all targets $t$.
\end{theorem}

Our algorithm works by simulating the 
``textbook'' Dynamic Programming algorithm of Bellman \cite{Bellman57} much faster, using ideas from linear sketching
to avoid recomputing target sums that are already known to be attainable.
We present a summary of these techniques in Section~\ref{techniques}.

An interesting feature of our algorithm is that, in contrast to most previous results on Subset Sum, 
it does not rely on the Fast Fourier Transform (FFT). 
In particular, by setting $m = s$ where $s$ is the sum of all input numbers,
our algorithm implies an algorithm for the non-Modular Subset Sum problem.
It matches the $\widetilde{O}(s)$ runtime achieved by Koiliaris and Xu \cite{KX17}, but without using FFT.

Another important property of our algorithm is that it does not need to know all the input numbers
in advance. Instead, it works in an \emph{online} fashion, by outputting all the newly attainable subset sums
for every new number that is provided.

Finally, the runtime of our algorithm is optimal, in the sense that there is no $n^{O(1)} m^{1-\epsilon}$ algorithm for any $\epsilon > 0$, 
assuming the Strong Exponential Time Hypothesis or other natural assumptions.
This is implied by 
recent results in Fine-Grained Complexity \cite{ALW14, CDLMNOPSW16, ABHS17} 
for the Subset Sum problem. Note that our algorithm matches this conditional lower bound for a single target $t$, while 
also outputting \emph{all} attainable subset sums.

We expect that our techniques will be applicable to other settings. In particular, an interesting open problem is the following:

\begin{theorem*}
Is there an $\widetilde{O}(n + M)$ time algorithm for the non-Modular Subset Sum problem, where
$M$ is the largest of the $n$ given integers?
\end{theorem*}
Such a runtime would improve the best known algorithm for Subset Sum \cite{Bringmann17} without contradicting any known conditional lower bounds.
In Section~\ref{discussion} we discuss how our techniques based on linear sketching could be helpful in resolving this question.

\subsection{Overview and techniques}
\label{techniques}

\paragraph{Certificate Complexity}
To illustrate our ideas, it is helpful to first consider the certificate complexity of the Modular Subset Sum problem.
It is easy to provide a certificate that target $t$ is attainable by just providing a list of elements that sum up to $t$.
But how can we certify that there is \emph{no} such subset?

An idea is to efficiently certify correctness of every step of Bellman's algorithm. 
Let $S^0 = \{0\}$ and $S^i$ be the set of attainable sums using the first $i$ integers. 
Bellman's algorithm computes $S^{i}$ as $S^{i-1} \cup (S^{i-1} + w_{i})$, where $w_i$ is the $i$-th integer.
The $O(nm)$ running time of Bellman's algorithm stems from the fact that every step costs $O(m)$ time, and there are $n$ steps.

To certify it more efficiently, the runtime of our algorithm shouldn't depend on the whole $S^i$, but rather spend time proportional to 
$\left|S^i \setminus S^{i-1}\right|$, i.e. the number of newly created sums. The certificate provides a set which is supposed
to be the set of newly added elements. While it is easy to certify that all elements from the provided set are indeed attainable,
the harder part is to certify that no elements are missing from the provided set. To do that, we perform 
Polynomial Identity Testing via inner products with random vectors to check if the characteristic vectors of two sets are the same.
To implement this efficiently, we show that it suffices to
use pseudorandom vectors obtained by linear hash functions. Linear hash functions allow one to compute very efficiently the
hash value of sets under shifts. This is important as operations of the form $S^{i-1} + w_{i}$ appear throughout the execution of Bellman's algorithm. We defer further details to Section~\ref{certificate}.

The ideas above suffice to obtain a non-deterministic algorithm running in time $\widetilde{O}(m)$ that guesses the newly created elements $S^i \setminus S^{i-1}$ at every step $i$ and certifies whether these guesses are correct. 
Removing the non-determinism and obtaining an actual algorithm is more challenging.
For many problems such as matrix multiplication
\cite{Freivalds77}, Orthogonal Vectors \cite{Williams16},
3-SUM, Linear Programming, and All-Pairs-Shortest-Paths \cite{CGIMPS16} 
the runtimes of the best algorithms are significantly worse than those of their non-deterministic counterparts. 
More generally, polynomial-sized certificates don't necessarily imply polynomial-time algorithms, 
as this is equivalent to the question ${\sf P} \overset{?}{=} {\sf NP \cap coNP}$.

\paragraph{Sketching}  

For the problem of Modular Subset Sum, however, we show that
using ideas from linear sketching it is possible to remove the 
non-determinism by incurring only a poly-logarithmic overhead in the runtime.

Specifically, besides just checking whether the characteristic vectors of two sets are the same,
we can use poly-logarithmic size linear sketches of the vectors to identify a \emph{position} in which they differ.
This works by carefully isolating elements by randomly subsampling subsets of entries of different sizes.
Naively computing all positions in which two sets differ would require computing new sketches many times with fresh randomness.
A contribution of our work is showing that only limited randomness is sufficient, which allows us to maintain only a few
data structures for evaluating the sketches.

This technique allows us to recover all elements from the symmetric difference between $S^{i-1}$ and $S^{i-1} + w_{i}$, in 
poly-logarithmic amortized time per element.
Observing that half of the elements in this symmetric difference are the newly attainable sums, we discover all of them
by spending time which is near-linear in the number of newly attainable sums.
More details can be found in Section~\ref{algorithm}.

Linear sketching has originally been developed with applications to streaming algorithms and dimensionality reduction.
Recently, it has also emerged as a powerful tool for 
Linear algebra \cite{Woodruff14},
dynamic graph algorithms \cite{AGM12,KKM13},
and approximation algorithms \cite{ANOY14}. 
However, to the best of our knowledge, our algorithm is one of the first 
applications of linear sketching to obtain fast algorithms for combinatorial problems in an offline setting.

%% file: preliminaries.tex
\section{Preliminaries}

We first define the problems of Subset Sum and Modular Subset Sum formally.

\begin{definition}[Subset Sum]
Given integers $w_1, \dots, w_n$ and a target $t$, decide whether there exists an $S\subseteq[n]$ such that $\sum\limits_{i\in S} w_i = t$.
\end{definition}

\begin{definition}[Modular Subset Sum]
Given integers $w_1, \dots, w_n\in\mathbb{Z}_m$ and a target $t\in\mathbb{Z}_m$, decide whether there exists an $S\subseteq[n]$ such that 
$\sum\limits_{i\in S} w_i \equiv t\ (mod\ m)$.
\end{definition}

We will also need the following notation.

\begin{definition}
Given $A\subseteq \mathbb{Z}_m$ and $x\in\mathbb{Z}_m$, we denote by $A+x = \{y + x\ |\ y\in A\}$ the operation of \emph{shifting} set $A$ by $x$.
\end{definition}

\begin{definition}
We denote by $\mathbb{Z}_m^*$ the set of integers in $\mathbb{Z}_m$ that are coprime with $m$. In other words,
\[ \mathbb{Z}_m^* = \{ x\in\mathbb{Z}_m\ |\ \gcd(x,m) = 1 \} \]
\end{definition}

\begin{definition}
Given a random variable $x$ and a probability distribution $\mathcal{D}$, $x\sim \mathcal{D}$ denotes that $x$ is sampled from $\mathcal{D}$.
Given a random variable $x$ and a set of outcomes $D$, $x\sim_U D$ denotes that $x$ is sampled uniformly at random from $D$.
\end{definition}
\begin{definition}
Given $x\in \mathbb{Z}_m$, we denote by $|x| = \min\{x, m - x\}$ the \emph{absolute value} of $x$.
\end{definition}

%% file: certificate.tex
\section{A Certificate for Bellman's algorithm}
\label{certificate}

To illustrate our ideas, we will provide an efficiently verifiable certificate for checking whether a given subset sum $t$ is attainable.
While certifying attainable subset sums is straightforward, certifying that no subset sums to $t$ is more challenging.
To achieve this, we provide a certificate that certifies the execution of Bellman's algorithm.
Even though Bellman's algorithm runs in $O(nm)$, we will show that it is possible to certify it in $\widetilde{O}(n + m)$ time.

Let $S^0 = \{0\}$ and $S^i$ be the set of attainable sums using the first $i$ integers. 
Bellman's algorithm computes $S^{i}$ as $S^{i-1} \cup (S^{i-1} + w_{i})$, where $w_i$ is the $i$-th integer.
To certify it more efficiently, the runtime of our algorithm shouldn't depend on the whole $S^i$, but rather spend time proportional to 
$\left|S^i \setminus S^{i-1}\right|$. To do this we certify all sets $S^i \setminus S^{i-1}$
of newly attainable subset sums after the $i$-th number is processed.

Given a collection of sets $N^i_{+}$ which are claimed to be $S^{i} \setminus S^{i-1}$, 
it is straightforward to
certify that $N^i_{+} \subseteq S^{i} \setminus S^{i-1}$ by checking for all $x\in N^i_{+}$ that $x - w_i \in S^{i-1}$ and $x\notin S^{i-1}$.
However, certifying that $N^{i}_{+}$ contains \emph{all} elements in $S^{i} \setminus S^{i-1}$ is significantly harder.
To perform this verification, suppose we are also given sets $N^i_{-}$, which are claimed to be equal to
$S^{i} \setminus (S^{i-1} + w_{i})$. 
Note that requiring knowledge of $N^i_{-}$ only doubles the valid certificate size, which directly follows from the claim below:
\begin{claim}
$|S^{i} \setminus S^{i-1}| = |S^{i} \setminus (S^{i-1} + w_{i})|$
\end{claim}

It is again simple to certify that $N^i_{-} \subseteq S^{i} \setminus (S^{i-1} + w_{i})$, by checking for all $x\in N^i_{-}$
that $x\in S^{i-1}$ and $x - w_i \notin S^{i-1}$.
To check that there are no elements missing from $N^{i}_+$ or $N^i_-$ we use the following claim, which follows from the above 
discussion.
\ifx 0 
, it is sufficient to check that
$N^{i}_{-} \dot{\cup} (S^{i-1} + w_i) = N^{i}_{+} \dot{\cup} S^{i-1}$, which must hold as both sides must be equal to $S_i$.
Equivalently, it is enough to certify that $v := \Ind_{N^i_{-}} + \Ind_{S^{i-1} + w_i} - \Ind_{N^i_+} - \Ind_{S^{i-1}} = \vec{0}$, where 
$\Ind_X \in \{0,1\}^m$ is the characteristic vector of set $X$. Note that $v\in \{-1,0,1\}^m$.

This directly implies the following:
\fi 

\begin{claim}
\label{n+-equiv}
Given that $N^i_{+} \subseteq S^{i} \setminus S^{i-1}$ and $N^i_{-} \subseteq S^{i} \setminus (S^{i-1} + w_{i})$,
the following statements are equivalent:
\begin{itemize}
\item $N^i_{+} = S^{i} \setminus S^{i-1}$ and $N^i_{-} = S^{i} \setminus (S^{i-1} + w_{i})$
\item $\Ind_{N^i_{-}} + \Ind_{S^{i-1} + w_i} - \Ind_{N^i_+} - \Ind_{S^{i-1}} = \vec{0}$
\end{itemize}
\end{claim}

By \cref{n+-equiv}, we just need to certify that the vector $\Ind_{N^i_{-}} + \Ind_{S^{i-1} + w_i} - \Ind_{N^i_+} - \Ind_{S^{i-1}}$
is the zero vector. A natural way to do this would be to use randomized identity testing. 
\begin{claim}
Given any non-zero vector $v\in\mathbb{R}^m$, 
we have $\Pr_{r\sim_U \{0,1\}^m} [\langle v, r \rangle \neq 0] \geq \frac{1}{2}$.
\end{claim}
In order to verify that for all $i$ the vector $v^i := \Ind_{N^i_{-}} + \Ind_{S^{i-1} + w_i} - \Ind_{N^i_+} - \Ind_{S^{i-1}}$
is the zero vector the idea is to sample a random vector $r$ and compute its inner products with $v^i$ for every
$i$. If any $v^i$ is non-zero, this process will detect the inconsistency with constant probability. By repeating, we can
amplify the probability of success.

Even though this would suffice, it would not be efficient. The reason for this is that we can't afford to explicitly keep
the characteristic vectors. Instead, we will directly maintain the inner products $\langle v^i , r\rangle$ by using an appropriate
data structure. Our data structure should allow us to compute $\langle \Ind_{S^{i-1} + w} , r \rangle$ efficiently for any given $w$.
However, for a random $r$ it seems necessary that such a data structure would need to spend significant time to re-compute the inner product
for any such $w$, when we move from $S^{i-1}$ to $S^i$. To alleviate this issue, instead of using uniformly random vectors, we will
choose a pseudo-random family that is more amenable to shifting. In particular, we will use the vector 
$r = \Ind_{\{i\ :\ ai+b\mod m\,\in\,[0,c]\}}$ for random $a, b, c$.

\begin{definition}[Pseudo-random distribution $\mathcal{D}$]\label{def:pseudorandom}
We define the distribution $\mathcal{D}$ of vectors\\
$r = \Ind_{\{i\ :\ ai+b\mod m\,\in\,[0,c]\}}$ where \\
$a\sim_U \mathbb{Z}_m^*$, \\
$b\sim_U \mathbb{Z}_m$, \\
$c\sim_U\{1,2,4,8,\dots,2^{\lceil\log m\rceil}\}$.
\end{definition}

We show that despite its limited randomness, this distribution can still be used for identity testing with a slightly smaller success
probability, which can again be amplified through repetition.

\begin{lemma}
\label{certlemma}
Given any non-zero vector $v\in\mathbb{R}^m$, 
we have $\Pr_{r\sim \mathcal{D}}[\langle v, r \rangle \neq 0] \geq \frac{1}{\widetilde{O}(\log m)}$.
\end{lemma}

The simplicity of the family of random vectors that we use allows us to efficiently update and compute inner products with
characteristic vectors of the form $\Ind_{S + w}$. Notice that $aj \mod m$ defines a permutation of the indices $j\in[m]$ since
$a\in\mathbb{Z}_m^*$. The computation below shows the effect of shifting $S$ by $w$.

\begin{align*}
\langle \Ind_{S + w}, r \rangle 
&= \sum\limits_{j=0}^{m-1} \Ind_{S+w} r_j 
= \sum\limits_{j=0}^{m-1} \Ind_{j\in S+w} \Ind_{aj + b \in [0,c]} 
= \sum\limits_{j=0}^{m-1} \Ind_{j\in S} \Ind_{a(j+w) + b \in [0,c]}\\ 
&= \sum\limits_{j=0}^{m-1} \Ind_{j\in S} \Ind_{aj \in [-b-aw,-b+c-aw]} 
= \sum\limits_{j=0}^{m-1} \Ind_{j\in a S} \Ind_{j \in [-b-aw,-b+c-aw]} \\
&= \sum\limits_{j\in[-b-aw,-b+c-aw]} \Ind_{j\in a S} 
\end{align*}
(Note that all the operations are in $\mathbb{Z}_m$)

From the above it becomes clear that 
the required inner product is equivalent to computing the number of elements of the set 
$aS = \{ax\ :\ x\in S\}$ that lie in a specified interval. To be able to efficiently compute these interval sums we use a
data structure that allows insertion of elements and range queries in logarithmic time. Such a data structure can be implemented 
using a Binary Search Tree. In order to update this data structure, one needs to insert all the ``permuted'' newly attainable subset
sums at step $i$, i.e. $a N^i_+$.

We will perform $2n$ queries (one for $S_{i-1}$ and one for $S_{i-1}+w_i$, for every step $i$), 
each taking $O(\log m)$ time. We will also perform at most $m$ insertions (one for every distinct subset sum), 
each taking $O(\log m)$ time. Thus the overall runtime is $\widetilde{O}(n + m)$ and succeeds with probability at least 
$1/\widetilde{O}(\log m)$, which can easily be amplified by repetition.

%% file: algorithm.tex
\section{From Certificate to Algorithm via Linear Sketching}
\label{algorithm}

So far we have seen how to efficiently check if $\Ind_{N^i_{-}} + \Ind_{S^{i-1} + w_i} - \Ind_{N^i_+} - \Ind_{S^{i-1}} = \vec{0}$ when provided a certificate listing the elements of $N^i_+$ and $N^i_-$. If the set $N^i_+$ only lists a strict subset of the elements of $S^i \setminus S^{i-1}$, our method would efficiently detect that. 
In order to make the process constructive, we want to be able to identify a missing element $x \in S^i \setminus S^{i-1} \setminus N^i_+$ in such cases. This way we can start from $N^i_+ = \emptyset$ and continue growing the set until we recover all elements. We would work similarly for recovering $N^i_-$. The main property that enables us to do so is summarized in the following claim which is an extension of Claim~\ref{n+-equiv}.

\begin{claim}
  Given that $N^i_{+} \subseteq S^{i} \setminus S^{i-1}$ and $N^i_{-} \subseteq S^{i} \setminus (S^{i-1} + w_{i})$, the indices of the positive non-zero entries of the vector
  $\Ind_{N^i_{-}} + \Ind_{S^{i-1} + w_i} - \Ind_{N^i_+} - \Ind_{S^{i-1}}$ correspond to the elements missing from $N^i_+$, i.e. $S^i \setminus S^{i-1} \setminus N^i_+$, while the indices of the negative non-zero entries correspond to the elements missing from $N^i_-$.
\end{claim}

\emph{Linear sketching} allows us to go beyond testing whether a vector is zero and identify the index of a non-zero element through inner products with 
carefully constructed random vectors. Consider a vector $v = (0,\dots,0,a,0,\dots,0)$ that only has a single non-zero entry at position $i^*$. One way to find its index is by multiplying by the all-ones vector to obtain $a = \langle v , \vec{1} \rangle$, and multiplying with the vector $Id = (1,2,\dots,m)$ to obtain
$a i^* = \langle v , Id \rangle$. Then, $i^*$ can be found by dividing the two values, $i^* = \frac {\langle v , Id \rangle} {\langle v , \vec{1} \rangle}$.
For an arbitrary vector $v$ containing $k$ non-zero entries, the same idea can be applied after randomly subsampling entries with probability $\approx \frac 1 k$ to isolate a single non-zero entry. For a randomly sampled set $R$, this corresponds to inner products with the vectors $\vec{1}_R$ and $Id_R$ where all entries outside $R$ are zeroed out. Linear sketching does not require knowledge of the sparsity parameter $k$ but constructs these random sets with different subsampling probabilities $2^{-1}, 2^{-2}, \dots , 2^{-\lceil\log m\rceil}$.

For our purposes, we show that the constructed sets $R$ need not be perfectly random but suffices to be pseudo-random. We will use the distribution of sets $R$ given by \cref{def:pseudorandom}. Such a set $R$ has the form $\{i\,:\, a i \in [l,r]\}$ for some interval $[l,r] = [-b,-b+c]$.  We define our sketch for some vector $v$ to be 
\begin{align*}
sketch_{a}(v, [l,r]) \triangleq \begin{pmatrix}  {\langle v , \vec{1}_R \rangle} \\ {\langle v , Id_R \rangle} \end{pmatrix} = \sum\limits_{i\,:\, a i \in [l,r]} 
\begin{pmatrix}v_i \\ iv_i\end{pmatrix}
\end{align*}

We show that this sketch recovers a non-zero entry of $v$ with non-trivial probability.

\begin{lemma}
\label{weaksketch}
Given any non-zero vector $v\in\mathbb{R}^m$, 
we have $\Pr_{r \sim \mathcal{D}}\left[v_i \neq 0 \text{ for } i = \frac {\langle v , Id_r \rangle} {\langle v , \vec{1}_r \rangle} \right] \geq \frac{1}{\widetilde{O}(\log m)}$.
\end{lemma}

Notice that again the sketch of a vector $v$ shifted by $w$, denoted as $v^{+w}$, can still be written in terms of a sketch of the original vector. 
\begin{align*}
sketch_{a}(v^{+w}, [l,r])
=& \sum\limits_{i\,:\,a i\in[l,r]} \begin{pmatrix}v_{i-w}\\ i v_{i-w}\end{pmatrix}
= \sum\limits_{i\,:\,a i\in[l-aw,r-aw]} \begin{pmatrix}v_{i}\\ (i+w) v_{i}\end{pmatrix}
= \sum\limits_{i\,:\,ai\in[l-aw,r-aw]} \begin{pmatrix}1 & 0 \\ w & 1\end{pmatrix} \begin{pmatrix}v_{i}\\ iv_{i}\end{pmatrix}\\
=& \begin{pmatrix}1 & 0 \\ w & 1\end{pmatrix} sketch_{a}(v,[l-aw,r-aw])
\end{align*}
Similar to \cref{certificate}, we can build a data structure for a given parameter $a$ to efficiently compute the sketch for any range $[l,r]$. See further details in Appendix~\ref{sec:datastructure}.

\paragraph{Recovering multiple non-zeros} Our discussion so far has focused on identifying a single non-zero entry of the vector $v^i = \Ind_{N^i_{-}} + \Ind_{S^{i-1} + w_i} - \Ind_{N^i_+} - \Ind_{S^{i-1}}$. However, once a sketch is used to find a single non-zero, it won't give any additional non-zero entries.
To recover more entries, we need a new sketch with fresh randomness which can be expensive to compute and maintain. Computing the sketch $sketch_{a}(v, [l,r])$ for different parameters $a$ would require rebuilding a new data structure from scratch. On the other hand, computing the sketch for a different interval $[l,r]$ 
but the same parameter $a$ can be efficiently performed using a single data structure.

We show that this is sufficient to recover a constant fraction of the non-zeros. 
For a given parameter $a$, we can compute sketches for disjoint windows, each yielding a different index of a non-zero entry if the corresponding sketch is valid.

\begin{definition}[Valid sketch]
$sketch_a(v,[l,r]) = \begin{pmatrix} s_1 \\ s_2 \end{pmatrix}$ is \emph{valid} if it identifies an index $i = \frac{s_2}{s_1}$ for a non-zero element 
$v_i \neq 0$ in the corresponding set where $ai \in [l,r]$.
\end{definition}

We show that for a vector with $k$ non-zeros, computing sketches for all windows 
$[0, \ell-1], [\ell, 2 \ell-1],$ and so on, with a window size $\ell \approx \frac m k$, yields at least half of the non-zero elements of $v$ in expectation.

\begin{lemma}
\label{validsketches}
Given any non-zero vector $v\in\mathbb{R}^m$ with $k$ non-zero entries. For any $\ell \le \frac{m}{10 k \log\log m}$,
it holds that 
$$\E_{a \sim_U \mathbb{Z}_m^*}\Big[ \# \{j\in \{1,\dots, \left\lceil m/\ell\right\rceil\} \ :\ sketch_a(v,[(j-1)\ell, j\ell-1]) \text{ is valid} \} \Big] \geq \frac {k} {2}$$
\end{lemma}

By Markov's inequality, this means that a single data structure is sufficient to recover a constant fraction of the non-zero entries with constant probability.
Thus, with $\text{poly} \log (nm)$ data structures to compute $sketch_a$ for different values of $a$, we can find all non-zeros with high probability.

Even though by \cref{validsketches} we can take any window size $\ell \leq \frac{m}{10k\log\log m}$, the time needed to iterate over all windows is $O(\frac{m}{\ell})$ and so
we need to pick an $\ell \approx \frac{m}{k}$, so that we only spend time proportional to the sparsity of $v$.

\paragraph{Estimating the window size}
As mentioned before, we need to identify an appropriate window size $\ell$, in order to efficiently recover a constant fraction of non-zeros.
If we pick a window size that is too large, we will recover few or none.
On the other hand, if the window size is too small, the time spent iterating over $\lceil m/\ell\rceil$ windows will be much larger than $k$.
To identify an appropriate window size, we start with $\ell = 1$ and keep doubling until we find an $\ell$ for which a significant fraction of its windows yield valid sketches. 
We can estimate this fraction of windows that yield valid sketches through sampling.

\section{Main Algorithm}

To describe our algorithm, we will denote by $BST(a)$ an initially empty data structure that can efficiently maintain a vector $v\in\mathbb{Z}_m^m$, and 
allows changing entries of $v$ and computing $sketch_a(v,[l,r])$ for any given interval $[l,r]$. Given such a data structure $DS$, we denote by 
$DS.\textproc{Sketch}(l,r)$ this sketch for the corresponding vector $v$.
Both of these operations take $O(\log m)$ time.

We will use different parameters $a_j\in\mathbb{Z}_m^*$ for $j\in[L]$.
For each such $j$, we keep a data structure $DS_j$ which is initialized as $BST(a_j)$ and maintains vector $\Ind_{S^{i-1}}$ for every iteration $i$.

Furthermore, we will maintain a set $N^i = \{(z,+1)\ |\ z\in N^i_+\}\cup \{(z,-1)\ |\ z \in N^i_-\}$,
which keeps the elements of $N^i_+$ and $N^i_-$ with their corresponding sign. For each $i$, $N^i_+$ and $N^i_-$ are initialized as 
$\emptyset$ and grow as more elements are discovered, until $N^i_+ = S^i \setminus S^{i-1}$ and $N^i_- = S^i \setminus (S^{i-1} + w_i)$ respectively.
To be able to efficiently compute sketches of $N^i$, we keep a data structure $DN^i_j$ which is initialized as $BST(a_j)$ and maintains vector $\Ind_{N^i_+} - \Ind_{N^i_-}$
for every iteration $i$ and parameter $a_j$ with $j\in[L]$.

\begin{algorithm}[ht]
\caption{Finding all Modular Subset Sums}
\begin{algorithmic}[1]
\State Initialize $L=\widetilde{O}(\log(nm))$
\State Pick $a_1, \dots, a_L$ uniformly at random from $\mathbb{Z}_m^*$.
\For{$j=1\dots L$}
\State Initialize $DS_j$ as a data structure $\textproc{BST}(a_j)$ with element $(0,+1)$.
\EndFor
\State $S\leftarrow \{0\}$
\For{$i=1\dots n$}
	\State $N^i\leftarrow \emptyset$
	\For{$j=1\dots L$}
		\State Initialize $DN^i_j$ as a data structure $\textproc{BST}(a_j)$ with the elements of $N^i$.
		\State $\ell \leftarrow$ \textproc{EstimateWindowSize}$(a_j, DS_j,DN^i_j, w_i)$
		\If {$\ell \neq \bot$}
			\For{each window $B = [(p-1)\ell, p\ell-1]$}
				\State $N^i\leftarrow N^i \cup \textproc{FindNonZero}(a_j, DS_j, DN^i_j, B, w_i)$
			\EndFor
		\EndIf
	\EndFor
	\For {$(z,+1) \in N^i$}  
	\State $S\leftarrow S\cup \{z\}$
	\For{$j=1\dots L$}
			\State $DS_j.\textproc{insert}(z, +1)$
		\EndFor
	\EndFor
\EndFor
\State \Return $S$
\end{algorithmic}
\label{algo1}
\end{algorithm}

\subsection{Finding Non-Zero Elements}

In this section we describe the procedure that finds a non-zero element of the vector $v^i = \Ind_{N^i_{-}} + \Ind_{S^{i-1} + w_i} - \Ind_{N^i_+} - \Ind_{S^{i-1}}$ for a given window
while $w_i$ is being considered.
The first line of \cref{a:findnonzero} efficiently computes the value of $sketch_a(v^i,[l,r])$ as 
\[ sketch_a(\Ind_{S^{i-1}+w_i},[l - a w, r - aw]) - sketch_a(\Ind_{S^{i}}, [l,r]) - sketch_a(\Ind_{N^i_+} - \Ind_{N^i_-}, [l,r]) \]
using the available data structures. Line~3 finds the index of the non-zero entry, assuming that the sketch is valid. Line~4 evaluates the vector $v^i$ at the found index $z$ and Line~5
checks that the sketch was indeed valid.
\begin{algorithm}[ht]
\caption{Find Non-zero element in window}
\begin{algorithmic}[1]
\Function{FindNonZero}{$a, DS, DN, [l,r], w$}
	\State $\begin{pmatrix}s_1\\ s_2\end{pmatrix} \leftarrow 
	\begin{pmatrix}1 & 0 \\ w & 1\end{pmatrix} DS.\textproc{Sketch}(l - a w, r - aw) - DS.\textproc{Sketch}(l,r) - DN.\textproc{Sketch}(l,r)$
	\State $z \leftarrow s_2 / s_1$
	\State $sign \leftarrow DS[z-w] - DS[z] - DN[z]$
	\If {$sign \in\{-1, +1\} \textbf{ and } a z \in [l,r]$} 
		\State \Return $\left(z, sign\right)$
	\Else 
		\State \Return $\bot$
	\EndIf
\EndFunction
\end{algorithmic}
\label{a:findnonzero}
\end{algorithm}

\subsection{Estimating the Window size}

To identify an appropriate window size $\ell$, we start with $\ell = 1$ and keep doubling until we find an $\ell$ for which a significant fraction of its windows yield valid sketches. 
At every step, we estimate the fraction of windows that yield valid sketches by randomly sampling a window and checking whether the corresponding sketch is valid using 
\cref{a:findnonzero}.

\begin{algorithm}[ht]
\caption{Estimate Window Size}
\begin{algorithmic}[1]
\Function{EstimateWindowSize}{$a, DS, DN, w$}
	\For{$\ell$ \textbf{in} $\{2^0, 2^1, \dots, 2^{\lfloor\log m\rfloor}\}$}
		\State \textproc{cnt}$\leftarrow 0$
		\MRepeat \ $T = \widetilde{O}(\log (nm))$ times 
				\State Pick uniformly random window $B = [(p-1)\ell, p\ell-1]$
				\If {\textproc{FindNonZero}$(a, DS, DN, B, w) \neq \bot $}
				\State \textproc{cnt} $\leftarrow \textproc{cnt} + 1$
				\EndIf
		\EndRepeat
		\If {$\textproc{cnt} \geq \frac{T}{200\log\log m}$}
			\State \Return $\ell$
		\EndIf
	\EndFor
	\State \Return $\bot$
\EndFunction
\end{algorithmic}
\label{a:estimatewindowsize}
\end{algorithm}

\subsection{Analysis}
The goal of this section is to analyze the correctness and running time of \cref{algo1}, as summarized in the following theorem:

\begin{theorem}
\cref{algo1} runs in time $\widetilde{O}(n + m)$ and returns the set of attainable subset sums with high probability.
\end{theorem}

Recall that our goal is to recover all non-zero entries of $S^i \setminus S^{i-1}$. 
Given the current sets $N^{i}_+$ and $N^{i}_-$ we do that by identifying non-zero entries of the vector
\[ v^i = \Ind_{N^i_{-}} + \Ind_{S^{i-1} + w_i} - \Ind_{N^i_+} - \Ind_{S^{i-1}} \]

By Lemma~\ref{validsketches}
we know that for any $\ell \le \frac{m}{10 k \log\log m}$ it holds that 
$$\E_{a \sim_U \mathbb{Z}_m^*}\Big[ \# \{j\in \{1,\dots, \left\lceil m/\ell\right\rceil\} \ :\ sketch_a(v^i,[(j-1)\ell, j\ell-1]) \text{ is valid} \} \Big] \geq \frac {k} {2}$$
In particular, there exists an $\ell^*$ with $\frac{m}{20 k \log\log m} \leq \ell^* \leq \frac{m}{10 k \log\log m}$ 
for which the above holds and $\ell^*$ is a power of $2$.
By applying Markov's inequality we get that 
$$\Pr_{a \sim_U \mathbb{Z}_m^*}\Big[ \# \{j\in \{1,\dots, \left\lceil m/\ell^*\right\rceil\} \ :\ sketch_a(v^i,[(j-1)\ell^*, j\ell^*-1]) \text{ is valid} \} \geq \frac{k}{4} \Big] 
\geq \frac {1}{100\log\log m}$$

Whenever this happens for a chosen $a$, we call the parameter $a$ \emph{helpful} for vector $v^i$.

We argue that if a helpful $a$ is chosen, \textproc{EstimateWindowSize} will return an $\ell \le \ell^*$ that recovers at least $\frac k {40}$ of the non-zero entries of $v^i$. This follows from the two lemmas below:


\begin{lemma}\label{lem:density}
If \textproc{EstimateWindowSize} does not return $\bot$, it will return an $\ell$ with at least $\frac {m} {400 \ell \log \log m}$ valid sketches  with high probability.
\end{lemma}

\begin{lemma}\label{lem:density2}
  If the chosen parameter $a$ is helpful for $v^i$, then
\textproc{EstimateWindowSize} will return $\ell \le \ell^*$ with high probability.
\end{lemma}

The two lemmas imply that if the chosen parameter $a$ is helpful for vector $v^i$, \textproc{EstimateWindowSize} will return an $\ell \le \ell^*$ with 
at least $\frac {m} {400 \ell \log \log m} \ge \frac {m} {400 \ell^* \log \log m} \ge \frac k {40}$ valid sketches  with high probability.

Since $a$ is helpful with probability at least $\frac {1}{100\log\log m}$, among the $L = \widetilde{O}(\log (nm))$ data structures that we create with different parameters $a$, there will be at least $\log (nm)$ helpful data structures with high probability. This means that with probability $1 - \frac 1 { \text{poly}(nm) }$, all elements $S^i \setminus S^{i-1}$ and $S^i \setminus (S^{i-1} + w_i)$ will be recovered. The correctness follows by taking union bound over all $w_i$'s.

The runtime of the algorithm is $\widetilde{O}(n + m)$. Calls to \textproc{EstimateWindowSize} and \textproc{FindNonZero} only take $\text{poly} \log (nm)$ time. The number of calls to $\textproc{FindNonZero}$ depends on the window size $\ell$ returned by \textproc{EstimateWindowSize} at every iteration. Lemma~\ref{lem:density}, implies that whenever \textproc{EstimateWindowSize} does not return $\bot$, it returns an $\ell$ has at least $\frac {m} {400 \ell \log \log m}$ valid sketches. This means that even though $\frac m \ell$ calls to \textproc{FindNonZero} are made, the number of such calls can be upper-bounded by $400 \log \log m$ times valid sketches. There will be at most $m$ valid sketches throughout the execution of the algorithm which yields the claimed bound on the runtime.

%% file: discussion.tex
\section{Discussion}
\label{discussion}

\paragraph{On the optimality of our algorithm}

The work of~\cite{ABHS17} shows that the non-Modular Subset Sum problem cannot be solved in time $t^{1-\eps}\cdot 2^{o(n)}$ for any constant $\eps>0$ assuming the Strong Exponential Time Hypothesis (SETH). This lower bound also implies that the {\em Modular} Subset Sum problem cannot be solved in time $m^{1-\eps}\cdot 2^{o(n)}$. Indeed, suppose that there exists an $m^{1-\eps}\cdot 2^{o(n)}$ time algorithm for the Modular Subset Sum problem. Given an instance of the non-modular version of the problem, we set $m=t\cdot n$ and run the algorithm for the modular version. This solves the non-modular problem since we can assume that all given integers $w$ satisfy $w<t$ without loss of generality. We get $(tn)^{1-\eps}\cdot 2^{o(n)}=t^{1-\eps}\cdot 2^{o(n)}$ time algorithm for the non-modular version of the problem, which contradicts SETH (by~\cite{ABHS17}). Thus our algorithm is essentially optimal.

\paragraph{Runtime independent of $n$}
If the elements $w_i$ are succinctly described by listing the multiplicities of all elements in $[m]$, the runtime of the algorithm can be made $\widetilde{O}(m)$, independent of $n$. This is because once a given weight $w$ does not produce any new subset sums, we can ignore any other elements with the same weight again. Thus, the number of elements we will consider is $O(m)$. This is because the number of elements that will produce a new subset sum is at most $m$ and there will be at most $m$ times where a new element does not produce any subset sums.

\paragraph{Extension to higher dimensions}
Our algorithm can be easily extended to solve higher dimensional generalizations of Modular Subset Sum: 

\emph{Given a sequence of vectors $w_1, \dots, w_n\in\mathbb{Z}_m^d$ and a vector $t\in\mathbb{R}_m^d$, does there exist a subset $S$ of $[n]$ such that 
$\sum_{i\in I} w_i \equiv t\ (mod\ m)$?}

Our algorithm can list all attainable subset sums $t\in\mathbb{Z}_m^d$ in time $\widetilde{O}\left(m \right)^d$. Instead of using a data structure to compute interval sums and evaluate sketches, it needs to use a data structure supporting high dimensional range queries.

\paragraph{Applications to (non-Modular) Subset Sum}

The (non-Modular) Subset Sum can be seen as a special case of Modular Subset Sum, by setting $m = s$ where $s$ is the sum of all input numbers. 
Our algorithm implies an algorithm for this problem which matches the $\widetilde{O}(s)$ runtime achieved by Koiliaris and Xu \cite{KX17}, but interestingly does so without using FFT.

An interesting open question is what the best possible runtime for Subset Sum is. Even though the best known algorithm for Subset Sum \cite{Bringmann17} running in $\widetilde{O}(t)$ is optimal in its dependence on the target sum $t$ assuming the Strong Exponential Time Hypothesis, it is possible that an $\widetilde{O}(n + M)$ algorithm exists, where
$M$ is the largest of the $n$ given integers. Such an algorithm would not contradict any of the known conditional lower bounds.

\begin{theorem*}
Is there an $\widetilde{O}(n + M)$ time algorithm for the non-Modular Subset Sum problem, where
$M$ is the largest of the $n$ given integers?
\end{theorem*}

It is known that when elements are bounded by $M$, any instance of Subset Sum can be reduced to an instance where the target is in $[0, M]$ and all elements are possibly negative and lie in the range $[-M,M]$. This follows from~\cite{EW18} using Steinitz Lemma or from~\cite{Pisinger99} by the process of Balanced Fillings. In addition, for some ordering of the elements, any prefix of the optimal subset also lies within the same window $[-M,M]$. This is quite similar to Modular Subset Sum as one only needs to keep track of elements within an interval of size $O(M)$. It is an interesting question whether the ideas from linear sketching can be used to efficiently simulate Bellman's
algorithm in such a setting.

%% file: appendix.tex

\appendix 
\section{Missing proofs}

\subsection{Pseudo-random family}
The goal of this section is to prove Lemmas~\ref{certlemma}, \ref{weaksketch}, and \ref{validsketches}. We will first prove the following Lemma, which bounds the probability
that a specific non-zero element falls within the same window with some other non-zero.

\begin{lemma}
\label{helper}
Let $v$ be a vector in $\mathbb{R}^m$ with exactly $k$ non-zero elements. Then $\forall i\in[m]$,
\[ 
\underset{a\sim_U \mathbb{Z}_m^*}{\Pr}\left[\exists\ j\neq i\ : 
v_j \neq 0, 
\left|(ai - aj)\ mod\ m\right|
< \frac{m}{10k\log\log m}
\right] 
\leq \frac{1}{2} 
\]
\end{lemma}
\begin{proof}
First of all, note that 
\begin{align*}
&\underset{a\sim_U\mathbb{Z}_m^*}{\Pr}\left[\exists\ j\neq i\ : v_j\neq 0, \left|(ai - aj)\ mod\ m\right| < \frac{m}{10k\log\log m}\right] \\
\leq &
k \underset{a\sim_U\mathbb{Z}_m^*}{\Pr}\left[\left|at\ mod\ m\right|< \frac{m}{10k\log\log m} \right] 
\end{align*}
for a fixed $t\in\mathbb{Z}_m \backslash \{0\}$.
Now, if $gcd(t,m) = g > 1$, we can write $t = g t'$ and $m = g m'$, where $gcd(t',m') = 1$. But then we get
\begin{align*}
\underset{a\sim_U\mathbb{Z}_m^*}{\Pr}\left[|at\ mod\ m|< \frac{m}{10k\log\log m} \right] =
&\underset{a\sim_U\mathbb{Z}_m^*}{\Pr}\left[|agt'\ mod\ gm'|< \frac{gm'}{10k\log\log m} \right] \\=
&\underset{a\sim_U\mathbb{Z}_{m'}^*}{\Pr}\left[|agt'\ mod\ gm'|< \frac{gm'}{10k\log\log m} \right] \\=
&\underset{a\sim_U\mathbb{Z}_{m'}^*}{\Pr}\left[|at'\ mod\ m'|< \frac{m'}{10k\log\log m} \right]
\end{align*}
and so we can assume wlog that $gcd(t,m) = 1$.
Furthermore, note that in this case $t\mathbb{Z}_m^* = \mathbb{Z}_m^*$, and so wlog we can assume that $t=1$. 
So it is enough to bound
\begin{align*}
\underset{a\sim_U\mathbb{Z}_m^*}{\Pr}\left[|a| < \frac{m}{10k\log\log m} \right]
\leq
&\frac{2\frac{m}{10k\log\log m}}{|\mathbb{Z}_m^*|}\\
\leq &\frac{2\frac{m}{10k\log\log m}}{\frac{m}{2\log\log m}}\\
\leq&
\frac{1}{2k}
\end{align*}

Finally, we get 
\begin{align*}
&\underset{a\sim_U\mathbb{Z}_m^*}{\Pr}\left[\exists\ j\neq i\ : 
v_j\neq 0, \left|(ai - aj)\ mod\ m\right| < \frac{m}{10k\log\log m}\right] \\
&\leq k \underset{a\sim_U\mathbb{Z}_m^*}{\Pr}\left[|at\ mod\ m|< \frac{m}{10k\log\log m} \right] \\
&\leq k\frac{1}{2k} \\
&= \frac{1}{2}
\end{align*}
\end{proof}

\begin{prevproof}{Lemma}{certlemma}

  Let $S = \{i: v_i \neq 0\}$ be the set of coordinates of non-zero entries of $v$ and let $k=|S|$ be its size. 
  We argue that the set $R = \{i : (ai+b) \in [0,c]\}$ for parameters $a,b,c$ drawn according to the distribution $\mathcal{D}$ of Definition~\ref{def:pseudorandom} will contain a single element from $S$ with non-trivial probability.
  
  $$\Pr_{R \sim \mathcal{D}}[ |S \cap R| = 1 ] = 
  \sum_{i \in S} \Pr[ i \in R ] \Pr[ j \notin R \text{ for all }  j \in S\setminus\{i\} | i \in R ] $$
  
  With probability $\frac{ 1} {O(\log m)}$ the parameter $c$ is selected so that
  $\frac{m}{20k\log\log m} \le c \le \frac{m}{10k\log\log m}$. In that case, we have that
  $\Pr[ i \in R | c ] = \frac c m \ge \frac{1}{20k\log\log m}$.
  
  Moreover, by Lemma~\ref{helper}, we have that 
  $$\Pr[ j \notin R \text{ for all }  j \in S\setminus\{i\} | i \in R, c ] \ge 
  \underset{a\sim_U \mathbb{Z}_m^*}{\Pr}\left[\forall\ j\neq i\ : 
  v_j \neq 0, 
  \left|(ai - aj)\ mod\ m\right|
  > \frac{m}{10k\log\log m}
  \right] 
  > \frac{1}{2}
  $$
  
  Thus overall, 
  $$\Pr_{R \sim \mathcal{D}}[ |S \cap R| = 1 ] \ge  \frac{ 1} {O(\log m)} \cdot \frac{1}{40\log\log m} \ge \frac{ 1} {\widetilde{O}(\log m)}$$
\end{prevproof}

\begin{prevproof}{Lemma}{weaksketch}
  The result follows by the proof of Lemma~\ref{certlemma} which shows that with probability at least $\frac{ 1} {\widetilde{O}(\log m)}$ a random vector $r$ drawn from distribution $\mathcal{D}$ will isolate a single non-zero entry of $v$. In that case the sketch $\frac {\langle v , Id_r \rangle} {\langle v , \vec{1}_r \rangle}$ will give its index correctly. 
\end{prevproof}

\begin{prevproof}{Lemma}{validsketches}

To prove that
$$\E_{a \sim_U \mathbb{Z}_m^*}\Big[ \# \{j\in \{1,\dots, \left\lceil m/\ell\right\rceil\} \ :\ sketch_a(v,[(j-1)\ell, j\ell-1]) \text{ is valid} \} \Big] \geq \frac {k} {2}$$
it suffices to lower bound the probability that a given element is unique in its window of size $\ell$.

\begin{align*}
& \E_{a \sim_U \mathbb{Z}_m^*}\Big[ \# \{j\in \{1,\dots, \left\lceil m/\ell\right\rceil\} \ :\ sketch_a(v,[(j-1)\ell, j\ell-1]) \text{ is valid} \} \Big] \\
& \geq \sum\limits_{i\in\mathbb{Z}_m, v_i\neq 0} 
\underset{a\sim_U \mathbb{Z}_m^*}{\Pr}\left[\nexists\ j\neq i\ : v_j \neq 0, \left|(ai - aj)\ mod\ m\right| < \ell \right] \\
& \geq \sum\limits_{i\in\mathbb{Z}_m, v_i\neq 0} 
\underset{a\sim_U \mathbb{Z}_m^*}{\Pr}\left[\nexists\ j\neq i\ : v_j \neq 0, \left|(ai - aj)\ mod\ m\right| < \frac{m}{10k\log\log m} \right] \\
&\geq \frac {k} {2}\\
\end{align*}
where the last inequality follows by \cref{helper}.
\end{prevproof}

\subsection{Window size estimation}
Given parameter $a$ and a window size $\ell$, let $c_\ell$ be the number of valid sketches computed with parameter $a$ and windows of size $\ell$.
The total number of windows is $t_\ell = \lceil\frac{m}{\ell}\rceil$.

\begin{prevproof}{Lemma}{lem:density}
Suppose that an $\ell$ is returned such that $c_\ell < \frac{m}{400\ell\log\log m}$.
Since
\[ \mathbb{E}[\textproc{cnt}] = \frac{c_{\ell}T}{t_{\ell}} \]
Chernoff bounds yield
\begin{align*} 
\Pr\left[\textproc{cnt} \geq \frac{T}{200\log\log m}\right] 
=
&\Pr\left[\textproc{cnt} \geq \frac{t_\ell}{200c_\ell\log\log m} \mathbb{E}[\textproc{cnt}] \right] \\
\leq 
&e^{-\left(\frac{t_\ell}{200c_\ell\log\log m} - 1\right)\frac{c_\ell T}{t_\ell}/3} \\
\leq &
e^{-\frac{c_\ell T}{t_\ell}/3} \\
=&
e^{-\frac{T}{1200\log\log m}}\\ 
\leq &\frac{1}{(nm)^{10}} 
\end{align*}
where we used the fact that $T = \widetilde{\Theta}(\log(nm))$. Therefore with high probability $\textproc{cnt} < \frac{T}{200\log\log m}$, 
which contradicts
the assumption that such an $\ell$ will be returned.
\end{prevproof}

\begin{prevproof}{Lemma}{lem:density2}
Suppose to the contrary. This means that all $\ell<\ell^*$ were rejected by $\textproc{EstimateWindowSize}$ and at some point
$\ell=\ell^*$ was considered as the window size.
Now, the fact that $a$ is helpful means that $c_{\ell^*} \geq \frac{k}{4}$. Combining this with $\ell^* \geq \frac{m}{20k\log\log m} $ implies that
\[ \frac{c_{\ell^*}}{t_{\ell^*}} \geq \frac{k/4}{\lceil\frac{m}{\ell^*}\rceil} \geq \frac{1}{100\log\log m}\]
and so 
\[ \mathbb{E}[\textproc{cnt}] = \frac{c_{\ell^*}T}{t_{\ell^*}} \geq \frac{T}{100\log\log m}\]
By using Chernoff bounds this implies that
\[ 
 \Pr\left[\textproc{cnt} < \frac{T}{200\log\log m} \right] \leq e^{-\frac{T}{800 \log \log m}} \leq \frac{1}{(nm)^{10}}
 \]
where we used the fact that $T = \widetilde{\Theta}(\log (nm))$. Therefore with high probability $\ell^*$ will be returned, which is a contradiction.
We conclude that a window size $\ell\leq \ell^*$ will be returned.
\end{prevproof}

\section{Data Structure} \label{sec:datastructure}
Our data structure will basically be a Binary Search Tree, augmented with the ability to quickly compute sums of values in 
intervals of keys. 
We will denote the data structure by $DS_a$, where $a\in\mathbb{Z}_m^{*}$ is some parameter.

Now, given a set of index-value pairs $I \subseteq \mathbb{Z}_m\times \mathbb{Z}$, the Binary Search Tree
will contain the set of key-value pairs $\{(ai\ mod\ m, v)\ :\ (i,v)\in I\}$.

We will use this data structure for two different operations: Inserting a key-value pair, and querying for the
sum of values corresponding to an interval of keys. To be able to compute the queries fast, we will use
the standard technique of maintaining for each node of the BST the sum of values in its subtree.

In the following, let $n$ be the number of keys currently in the data structure.
\paragraph{Inserting}
Inserting into this data structure is identical to inserting in a Binary Search Tree, the only difference
being that we need to update the subtree value sums. This can be done very easily in $O(\log n)$ time,
by simply re-computing the sum of every node touched by summing the recursive values for its children and
the value of the node itself. Since any standard BST touches $O(\log n)$ nodes in this process, this is
our runtime bound.

\paragraph{Querying}
Now, we need to be able to query for the sum of values corresponding to keys in some interval $[l,r]$.
Note that here it might be the case that $l > r$, in which case we can break our interval into two intervals
$[l,m-1],[0,r]$, compute the answer for each one separately, and finally add them up. For 
some interval $[l,r]$ with $l\leq r$, we compute the sum in $O(\log n)$ time in a standard way, as it is
known that the answer can be written as the sum of $O(\log n)$ subtree sums.

%% file: main.bbl
\newcommand{\etalchar}[1]{$^{#1}$}
\begin{thebibliography}{ANOY14}

\bibitem[ABHS17]{ABHS17}
Amir Abboud, Karl Bringmann, Danny Hermelin, and Dvir Shabtay.
\newblock Seth-based lower bounds for subset sum and bicriteria path.
\newblock {\em arXiv preprint arXiv:1704.04546}, 2017.

\bibitem[AGM12]{AGM12}
Kook~Jin Ahn, Sudipto Guha, and Andrew McGregor.
\newblock Analyzing graph structure via linear measurements.
\newblock In {\em Proceedings of the twenty-third annual ACM-SIAM symposium on
  Discrete Algorithms}, pages 459--467. Society for Industrial and Applied
  Mathematics, 2012.

\bibitem[Alo87]{Alon87}
Noga Alon.
\newblock Subset sums.
\newblock {\em Journal of Number Theory}, 27(2):196--205, 1987.

\bibitem[ALW14]{ALW14}
Amir Abboud, Kevin Lewi, and Ryan Williams.
\newblock Losing weight by gaining edges.
\newblock In {\em European Symposium on Algorithms}, pages 1--12. Springer,
  2014.

\bibitem[ANOY14]{ANOY14}
Alexandr Andoni, Aleksandar Nikolov, Krzysztof Onak, and Grigory Yaroslavtsev.
\newblock Parallel algorithms for geometric graph problems.
\newblock In {\em Proceedings of the forty-sixth annual ACM symposium on Theory
  of computing}, pages 574--583. ACM, 2014.

\bibitem[AT18]{AT18}
Kyriakos Axiotis and Christos Tzamos.
\newblock Capacitated dynamic programming: Faster knapsack and graph
  algorithms.
\newblock {\em arXiv preprint arXiv:1802.06440}, 2018.

\bibitem[Bel57]{Bellman57}
Richard Bellman.
\newblock Dynamic programming (dp).
\newblock 1957.

\bibitem[BHSS18]{BHSS18}
MohammadHossein Bateni, MohammadTaghi HajiAghayi, Saeed Seddighin, and Clifford
  Stein.
\newblock Fast algorithms for knapsack via convolution and prediction.
\newblock STOC, 2018.

\bibitem[Bri17]{Bringmann17}
Karl Bringmann.
\newblock A near-linear pseudopolynomial time algorithm for subset sum.
\newblock In {\em Proceedings of the Twenty-Eighth Annual ACM-SIAM Symposium on
  Discrete Algorithms}, pages 1073--1084. Society for Industrial and Applied
  Mathematics, 2017.

\bibitem[CDL{\etalchar{+}}16]{CDLMNOPSW16}
Marek Cygan, Holger Dell, Daniel Lokshtanov, D{\'a}niel Marx, Jesper Nederlof,
  Yoshio Okamoto, Ramamohan Paturi, Saket Saurabh, and Magnus Wahlstr{\"o}m.
\newblock On problems as hard as cnf-sat.
\newblock {\em ACM Transactions on Algorithms (TALG)}, 12(3):41, 2016.

\bibitem[CGI{\etalchar{+}}16]{CGIMPS16}
Marco~L Carmosino, Jiawei Gao, Russell Impagliazzo, Ivan Mihajlin, Ramamohan
  Paturi, and Stefan Schneider.
\newblock Nondeterministic extensions of the strong exponential time hypothesis
  and consequences for non-reducibility.
\newblock In {\em Proceedings of the 2016 ACM Conference on Innovations in
  Theoretical Computer Science}, pages 261--270. ACM, 2016.

\bibitem[CIO16]{CIO15}
Jean Cardinal, John Iacono, and Aur{\'e}lien Ooms.
\newblock Solving k-sum using few linear queries.
\newblock In {\em 24th Annual European Symposium on Algorithms, ESA 2016}.
  Schloss Dagstuhl-Leibniz-Zentrum fur Informatik GmbH, Dagstuhl Publishing,
  2016.

\bibitem[EGZ61]{EGZ61}
Paul Erdos, Abraham Ginzburg, and Abraham Ziv.
\newblock Theorem in the additive number theory.
\newblock {\em Bull. Res. Council Israel F}, 10:41--43, 1961.

\bibitem[ES16]{ES16}
Esther Ezra and Micha Sharir.
\newblock The decision tree complexity for $ k $-sum is at most nearly
  quadratic.
\newblock {\em arXiv preprint arXiv:1607.04336}, 2016.

\bibitem[EW18]{EW18}
Friedrich Eisenbrand and Robert Weismantel.
\newblock Proximity results and faster algorithms for integer programming using
  the steinitz lemma.
\newblock In {\em Proceedings of the Twenty-Ninth Annual ACM-SIAM Symposium on
  Discrete Algorithms}, pages 808--816. SIAM, 2018.

\bibitem[Fre77]{Freivalds77}
Rusins Freivalds.
\newblock Probabilistic machines can use less running time.
\newblock In {\em IFIP congress}, volume 839, page 842, 1977.

\bibitem[HLS08]{HLS07}
YO~Hamidounea, AS~Llad{\'o}b, and O~Serrab.
\newblock On complete subsets of the cyclic group.
\newblock {\em Journal of Combinatorial Theory, Series A}, 115:1279--1285,
  2008.

\bibitem[Kar72]{Karp72}
Richard~M Karp.
\newblock Reducibility among combinatorial problems.
\newblock In {\em Complexity of computer computations}, pages 85--103.
  Springer, 1972.

\bibitem[KKM13]{KKM13}
Bruce~M Kapron, Valerie King, and Ben Mountjoy.
\newblock Dynamic graph connectivity in polylogarithmic worst case time.
\newblock In {\em Proceedings of the twenty-fourth annual ACM-SIAM symposium on
  Discrete algorithms}, pages 1131--1142. Society for Industrial and Applied
  Mathematics, 2013.

\bibitem[KLM18]{KLM18}
Daniel~M Kane, Shachar Lovett, and Shay Moran.
\newblock Near-optimal linear decision trees for k-sum and related problems.
\newblock In {\em Proceedings of the 50th Annual ACM SIGACT Symposium on Theory
  of Computing}, pages 554--563. ACM, 2018.

\bibitem[KX17]{KX17}
Konstantinos Koiliaris and Chao Xu.
\newblock A faster pseudopolynomial time algorithm for subset sum.
\newblock In {\em Proceedings of the Twenty-Eighth Annual ACM-SIAM Symposium on
  Discrete Algorithms}, pages 1062--1072. SIAM, 2017.

\bibitem[LN10]{LN10}
Daniel Lokshtanov and Jesper Nederlof.
\newblock Saving space by algebraization.
\newblock In {\em Proceedings of the forty-second ACM symposium on Theory of
  computing}, pages 321--330. ACM, 2010.

\bibitem[MadH84]{Meyer84}
Friedhelm Meyer auf~der Heide.
\newblock A polynomial linear search algorithm for the n-dimensional knapsack
  problem.
\newblock {\em Journal of the ACM (JACM)}, 31(3):668--676, 1984.

\bibitem[Ols68]{Olson68}
John~E Olson.
\newblock An addition theorem modulo p.
\newblock {\em Journal of Combinatorial Theory}, 5(1):45--52, 1968.

\bibitem[Ols75]{Olson75}
John Olson.
\newblock Sums of sets of group elements.
\newblock {\em Acta Arithmetica}, 2(28):147--156, 1975.

\bibitem[Pfe99]{Pferschy99}
Ulrich Pferschy.
\newblock Dynamic programming revisited: Improving knapsack algorithms.
\newblock {\em Computing}, 63(4):419--430, 1999.

\bibitem[Pis99]{Pisinger99}
David Pisinger.
\newblock Linear time algorithms for knapsack problems with bounded weights.
\newblock {\em Journal of Algorithms}, 33(1):1--14, 1999.

\bibitem[Pis03]{Pisinger03}
David Pisinger.
\newblock Dynamic programming on the word ram.
\newblock {\em Algorithmica}, 35(2):128--145, 2003.

\bibitem[Sze70]{Szemeredi70}
Endre Szemer{\'e}di.
\newblock On a conjecture of erd{\"o}s and heilbronn.
\newblock {\em Acta Arithmetica}, 17(3):227--229, 1970.

\bibitem[Vu08]{Vu08}
Van Vu.
\newblock A structural approach to subset-sum problems.
\newblock In {\em Building Bridges}, pages 525--545. Springer, 2008.

\bibitem[Wil16]{Williams16}
Richard~Ryan Williams.
\newblock Strong eth breaks with merlin and arthur: Short non-interactive
  proofs of batch evaluation.
\newblock In {\em 31st Conference on Computational Complexity}, 2016.

\bibitem[Woo14]{Woodruff14}
David~P Woodruff.
\newblock Sketching as a tool for numerical linear algebra.
\newblock {\em Foundations and Trends{\textregistered} in Theoretical Computer
  Science}, 10(1--2):1--157, 2014.

\end{thebibliography}
